\documentclass[10pt,journal,compsoc]{IEEEtran}

\ifCLASSOPTIONcompsoc
  \usepackage[nocompress]{cite}
\else
  \usepackage{cite}
\fi

\ifCLASSINFOpdf
  \usepackage[pdftex]{graphicx}
\else
\fi

\usepackage{amsmath,amssymb,varioref}
\usepackage{amsthm}
\usepackage{mathtools}
\usepackage{graphicx}
\usepackage{xspace}

\usepackage{algorithm}
\usepackage{complexity}
\newcommand{\xEQP}{\EQP\xspace}
\newcommand{\xP}{\P\xspace}
\newcommand{\EQPk}{$\text{\EQP}_K$\xspace}
\newcommand{\EQPC}{$\text{\EQP}_\mathbb{C}$\xspace}

\newcounter{mycounter}                                                                                                  

\newenvironment{noindqlist} {\begin{list}{$\ket{\arabic{mycounter}}$~~}{\usecounter{mycounter} \labelsep=0em \labelwidth=0em \leftmargin=0em \itemindent=0em}} {\end{list}}

\newtheorem{defn}{Definition}
\newtheorem{lemma}{Lemma}

\newcounter{TableCounter}

\providecommand{\jacobi}[2]{\left( \displaystyle\frac{#1}{#2} \right)}
\providecommand{\Z}{\mathbb{Z}}
\providecommand{\Zp}{\Z/p\Z}

\providecommand{\ket}[1]{\left| #1 \right\rangle}

\usepackage{geometry}

\begin{document}
\title{Quantum Computers Can Find Quadratic Nonresidues in
  Deterministic Polynomial Time}

\author{Thomas~G.~Draper
\IEEEcompsocitemizethanks{\IEEEcompsocthanksitem T. Draper is with the Center for Communications Research at La Jolla,
4320 Westerra Court, San Diego, CA, 92121.\protect\\
E-mail: tdraper@ccrwest.org}
}

\IEEEtitleabstractindextext{%
\begin{abstract}
An integer $a$ is a quadratic nonresidue for a prime $p$ if $x^2 \equiv a \bmod p$ has no solution.
Quadratic nonresidues may be found by probabilistic methods in polynomial time.
However, without assuming the Generalized Riemann Hypothesis, no deterministic polynomial-time algorithm is known.
We present a quantum algorithm which generates a random quadratic nonresidue in deterministic polynomial time.
\end{abstract}

}

\maketitle
\IEEEdisplaynontitleabstractindextext
\IEEEpeerreviewmaketitle
\IEEEraisesectionheading{\section{Introduction}\label{sec:introduction}}

\IEEEPARstart{F}{inding} a quadratic nonresidue modulo a prime $p$ is an easy task.
A random choice of $x \in \{1, \ldots , p-1\}$ has a 50\% chance of being a quadratic nonresidue.
If $p$ is an $n$-bit prime, the Legendre symbol 
can be computed using $\log n$ multiplications, and thus in $O(M(n)\log
n)=O(n\log^2 n)$ time ~\cite{DBLP:journals/corr/abs-1004-2091,harvey:hal-02070778}.
The result certifies whether or not $x$ is a quadratic nonresidue.
Surprisingly, there is no known method for deterministically generating a quadratic nonresidue in polynomial time.
Assuming the Generalized Riemann Hypothesis, there exists a quadratic nonresidue less than $O(\log^2 p)$~\cite[p. 34]{1993--cohen}, that can therefore be found deterministically in $O(n^3\log^2 n)$ time by incremental search.
This paper presents a quantum algorithm that finds a quadratic nonresidue in $O(n\log^2 n)$ deterministic time, independent of the Riemann Hypothesis.

\section{Quadratic Residues}
\begin{defn}
  Let $a,p \in \Z$ where $gcd(a,p)=1$.
  If $x^2 \equiv a \bmod{p}$ has a solution, then $a$ is a \emph{quadratic residue} modulo $p$.
  Otherwise, $a$ is a \emph{quadratic nonresidue} modulo $p$.                                                           \end{defn}
                                                                                                                        \begin{defn}                                                                                                              Let $p$ be an odd prime and $a\in \Z$.                                                                                  The {\bf Legendre symbol} is defined as
  \begin{equation}                                                                                                          \jacobi{a}{p}=                                                                                                        \begin{cases}                                                                                                             1&\text{if }a\text{ is a quadratic residue}\bmod p,\\
    -1&\text{if }a\text{ is a quadratic nonresidue}\bmod p,\\                                                               0&\text{if }a\equiv 0\bmod{p}.\\                                                                                      \end{cases}                                                                                                           \end{equation}                                                                                                          \end{defn}

Recall the following facts about quadratic residues for an odd prime $p$:
\begin{equation}\label{pos}
    \left|\left\{a\in\Zp : \jacobi{a}{p}=1\right\}\right|= \frac{p-1}{2}
  \end{equation}
\begin{equation}\label{neg}
    \left|\left\{a\in\Zp : \jacobi{a}{p}=-1\right\}\right|= \frac{p-1}{2}
  \end{equation}
  \begin{equation}\label{qnr1}
    \jacobi{-1}{p}=(-1)^{\frac{p-1}{2}}=
      \begin{cases}
        1 & \text{if } p\equiv 1\bmod{4}\\
        -1 & \text{if } p\equiv 3\bmod{4}\\
      \end{cases}
  \end{equation}
  \begin{equation}\label{qnr2}
    \jacobi{2}{p}=(-1)^{\frac{p^2-1}{8}}=
      \begin{cases}
        1 & \text{if } p\equiv 1,7\bmod{8}\\
        -1 & \text{if } p\equiv 3,5\bmod{8}\\
      \end{cases}
  \end{equation}

Using equations (\ref{qnr1}) and (\ref{qnr2}), we see that $-1$ or $2$ is a quadratic nonresidue unless $p\equiv 1\bmod{8}$.
Thus the difficult primes for finding quadratic nonresidues are congruent to $1$ modulo $8$.

\begin{lemma}
  Let $p$ be prime such that $p \equiv 1 \bmod{4}$.
  If we take $1,2,\ldots,p-1$ as our nonzero congruence class representatives,
  then half of the quadratic nonresidues even and half are odd.
\end{lemma}

\begin{proof}
  By (\ref{qnr1}), $-1$ is a quadratic residue when $p \equiv 1 \bmod{4}$.
  Thus if $x$ is a quadratic nonresidue, then $-x\equiv p-x\bmod{p}$ is also a quadratic nonresidue.
  Since $p$ is odd, $x$ and $p-x$ have different parities, and $x\neq p-x$.
  Therefore, every odd quadratic nonresidue has a unique matching even nonresidue.
  Using (\ref{neg}), the number of nonresidues is $\frac{p-1}{2}$.
  Thus the number of odd (even) nonresidues is $\frac{p-1}{4}$.
\end{proof}

\section{Grover's Algorithm}
In 1996, Lov Grover~\cite{Grover96} devised an unstructured search algorithm using amplitude amplification. In the case of finding a single marked entry out of $N$, Grover's algorithm reduces the number of black box queries from $O(N)$ to $O(\sqrt{N})$. When the acceptable answer space has size $k$, the expected number of black box queries drops from $O(N/k)$ to $O(\sqrt{N/k})$.

Each step of Grover's algorithm involves flipping the signs of the amplitudes of the marked states and then inverting about the mean.

\subsection{Inversion about the mean}
Although Grover's algorithm uses only real amplitudes, the inversion about the mean holds for complex numbers and is a 2D inversion about a point.
Let $\bar{\alpha}$ be the mean of the amplitudes of a given state.
\[\bar{\alpha} = \frac{1}{2^n}\sum_{x=0}^{2^n-1}\alpha_x.\]
The inversion step of Grover's algorithm proceeds as follows.
\begin{enumerate}
  \item Initial state before inversion
    \[ \displaystyle\sum_{x=0}^{2^n-1}\alpha_x\ket{x}
    \]
  \item Quantum Hadamard Transform
    \[ \displaystyle\sum_{x=0}^{2^n-1}\frac{\alpha_x}{2^{n/2}}\sum_{y=0}^{2^n-1}(-1)^{x\cdot{}y}\ket{y}
    \]
  \item Negate $\ket{0}$
    \[
      \displaystyle{\left(\sum_{x=0}^{2^n-1}\frac{\alpha_x}{2^{n/2}}\sum_{y=0}^{2^n-1}(-1)^{x\cdot{}y}\ket{y}\right)-\sum_{x=0}^{2^n-1}\frac{\alpha_x}{2^{n/2}}\cdot{}2\ket{0}}
    \]
  \item Quantum Hadamard Transform
  \[
    \displaystyle\sum_{x=0}^{2^n-1}\alpha_x\ket{x}-2\bar{\alpha}\sum_{x=0}^{2^n-1}(-1)^{x\cdot{}0}\ket{x}
    \]
  \item Change global phase by $-1$ (NOP)
  \[
    \displaystyle\sum_{x=0}^{2^n-1}\left(2\bar{\alpha}-\alpha_x\right)\ket{x}
    \]
\end{enumerate}

Thus inversion about the mean maps each amplitude
\[\alpha_x \mapsto 2\bar{\alpha}-\alpha_x.\]
In using a inversion about the mean to find quadratic nonresidues, we will use complex amplitudes.

\subsection{Example: Grover's algorithm on four states}
Considering the special case of searching 4 states for a single marked state, only one Grover iteration is needed.
\begin{enumerate}
\item Initialize register in an equal superposition.

  \begin{figure}[h]
\centering
  \includegraphics[width=2.5in]{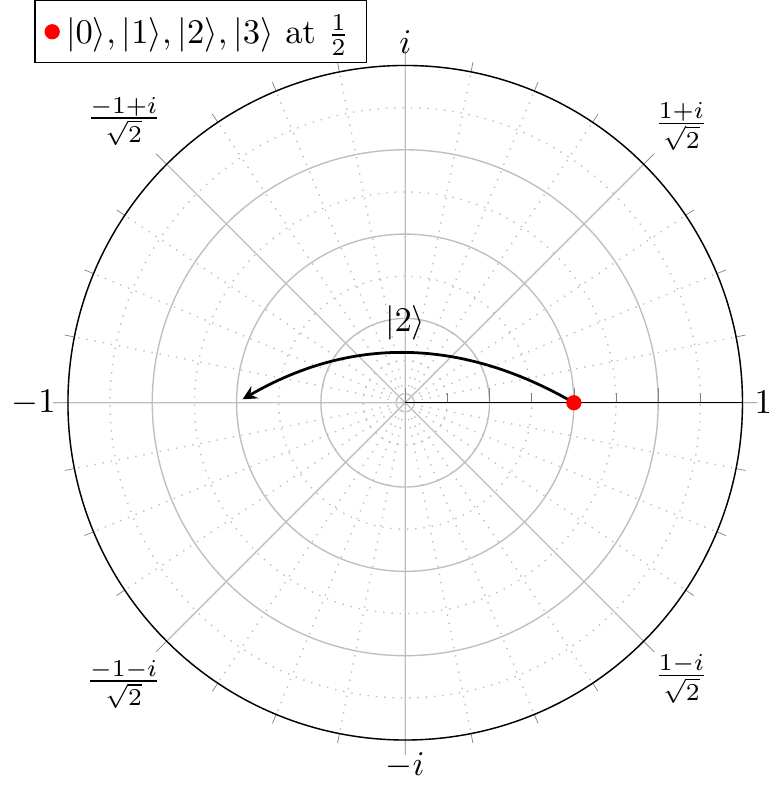}
  \caption{From equal superposition, negate $\ket{2}$ state.}
\label{grover_1}
\end{figure}

\item Negate phase according to indicator function (Figure~\ref{grover_1}). In this case, $\ket{2}$ is the target state.

  \begin{figure}[h]
\centering
  \includegraphics[width=2.5in]{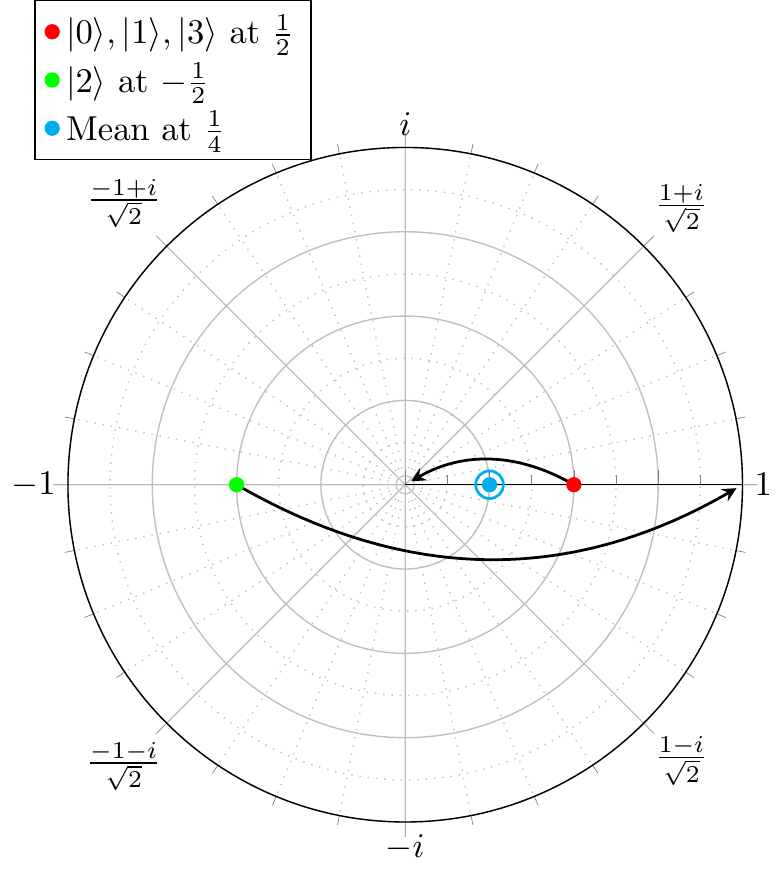}
    \caption{invert about the mean.}
\label{grover_2}
\end{figure}

\item invert all amplitudes about the mean of $1/4$ (Figure~\ref{grover_2}).

  \begin{figure}[h]
\centering
  \includegraphics[width=2.5in]{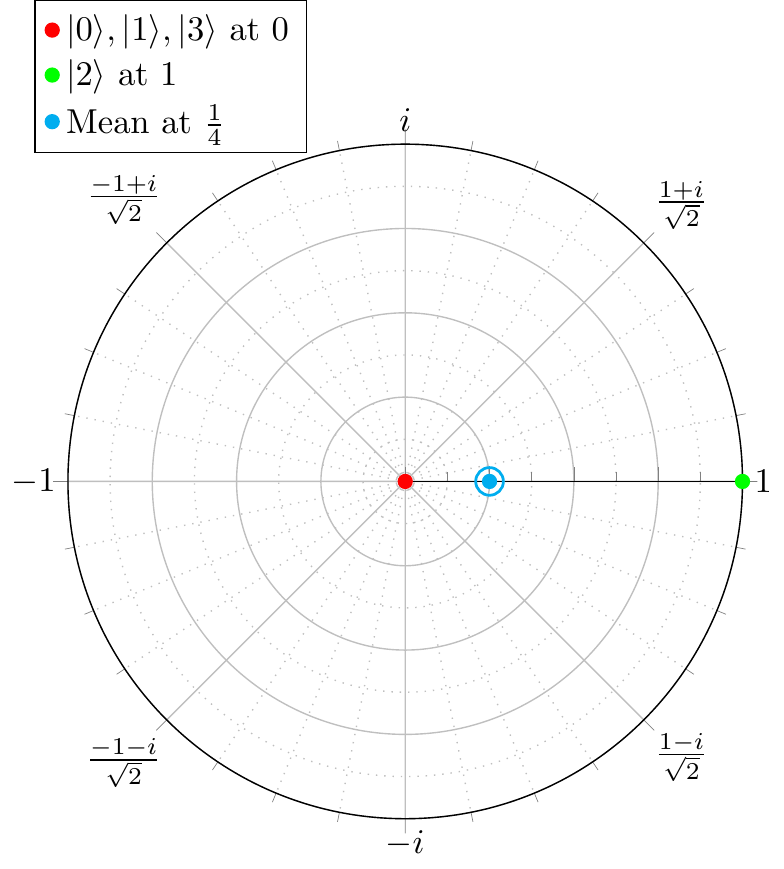}
    \caption{Final amplitudes, where only $\ket{2}$ is nonzero.}
\label{grover_3}
\end{figure}
    \item {After inversion, only $\ket{2}$ has nonzero amplitude.}
\end{enumerate}

At this point, all of the probability amplitude is associated with $\ket{2}$ and
the other three states have zero chance of being observed (Figure~\ref{grover_3}).
Observation at this point yields $\ket{2}$ with probability $1$.

\section{Amplify Quadratic Nonresidues}

Let $n$ be the least integer such that $2^n>p$ and let $N=2^n$.
By varying the Grover step slightly, we may deterministically generate quadratic nonresidues for $p\equiv 1\bmod{8}$.
Instead of negating the phase, we will rotate the complex phase of the even and
odd quadratic nonresidues in such a way that the mean moves to a position that
will send all non-nonresidues ($\ket{0}$, residues and values greater than or equal to $p$  but less than $N$) to zero upon inversion.
It is important to restrict the rotation to quadratic nonresidues less than $p$.
The distribution of residues and nonresidues from $p$ to $2^n-1$ is unknown and is not guaranteed to be equally partitioned.

Note that after the inversion, all nonresidues will have equal probability of being seen, and that we have no control over which one we observe.

\subsection{Example: Amplify the nonresidues for \protect{\boldmath{$p=41$}}, \protect{\boldmath{$N=64$}}}
In this example, we will use QNR to denote {\it quadratic nonresidues}.
\begin{enumerate}
  \item Initialize all $64$ values ($\ket{0},\ldots ,\ket{63}$) in an equal superposition.
    Since all values have an amplitude of $1/8$, the target mean will be $1/16$.
    Since $p=41$, there are $20$ QNRs, of which $10$ are even and $10$ are odd.

\begin{figure}[h]
\centering
  \includegraphics[width=2.5in]{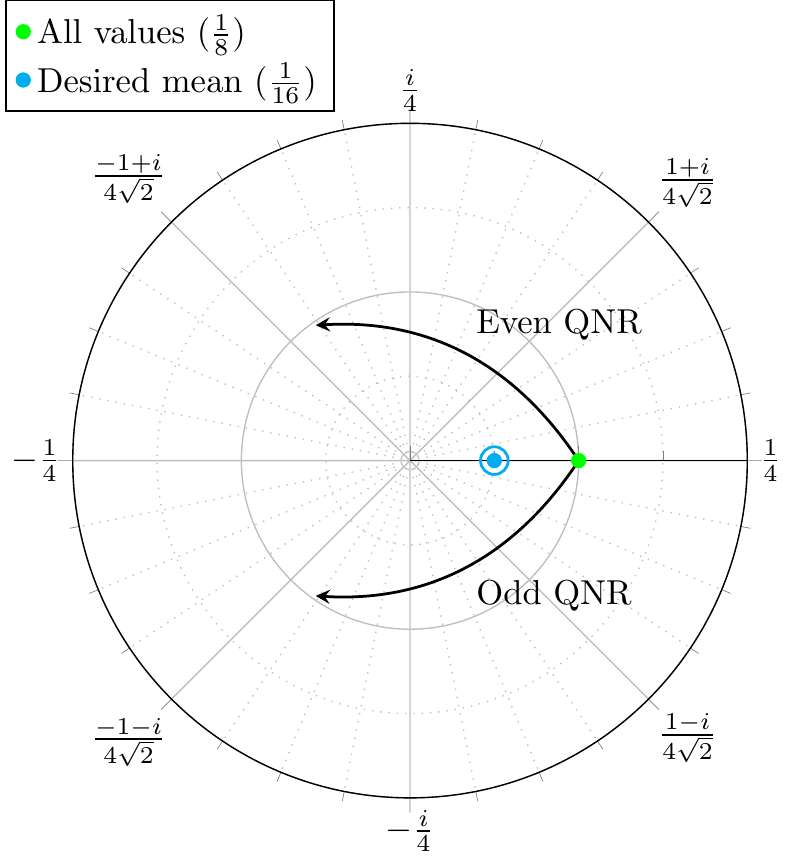}
    \caption{Rotate QNRs, even by $\theta$ and odd by $-\theta$.}
\label{p41_1}
\end{figure}

\item Let $\theta=\hbox{arccos}(-3/5)$. The even QNRs will be rotated by $\theta$ and the odd QNRs will be rotated by $-\theta$ (Figure~\ref{p41_1}).
  After rotation, the even/odd QNR amplitudes are $-\frac{3}{40}\pm\frac{i}{10}$.
  In calculating the mean, the imaginary components of the even and odd
  quadratic nonresidues will exactly cancel each other.
  The remaining average on the real line is thus
  \[\frac{1}{64}\left(20\cdot(-\frac{3}{40})+44\cdot\frac{1}{8}\right).\]
  Therefore, the mean is $\frac{1}{16}$.

\begin{figure}[h]
\centering
  \includegraphics[width=2.5in]{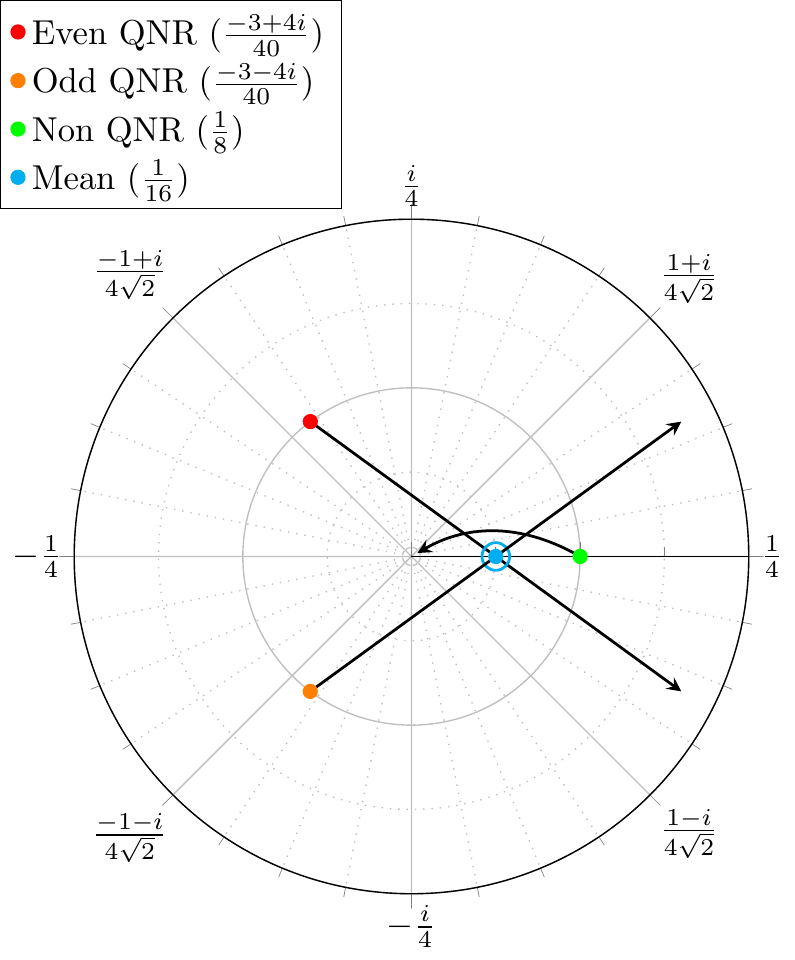}
    \caption{Invert about the mean.}
\label{p41_2}
\end{figure}
\item Invert all complex amplitudes about the mean. Each amplitude $a$ will invert to $2(\frac{1}{16})-a$ (Figure~\ref{p41_2}). In particular, $\frac{1}{8}$ inverts to zero, and the other amplitudes do not.

\begin{figure}[h]
\centering
  \includegraphics[width=2.5in]{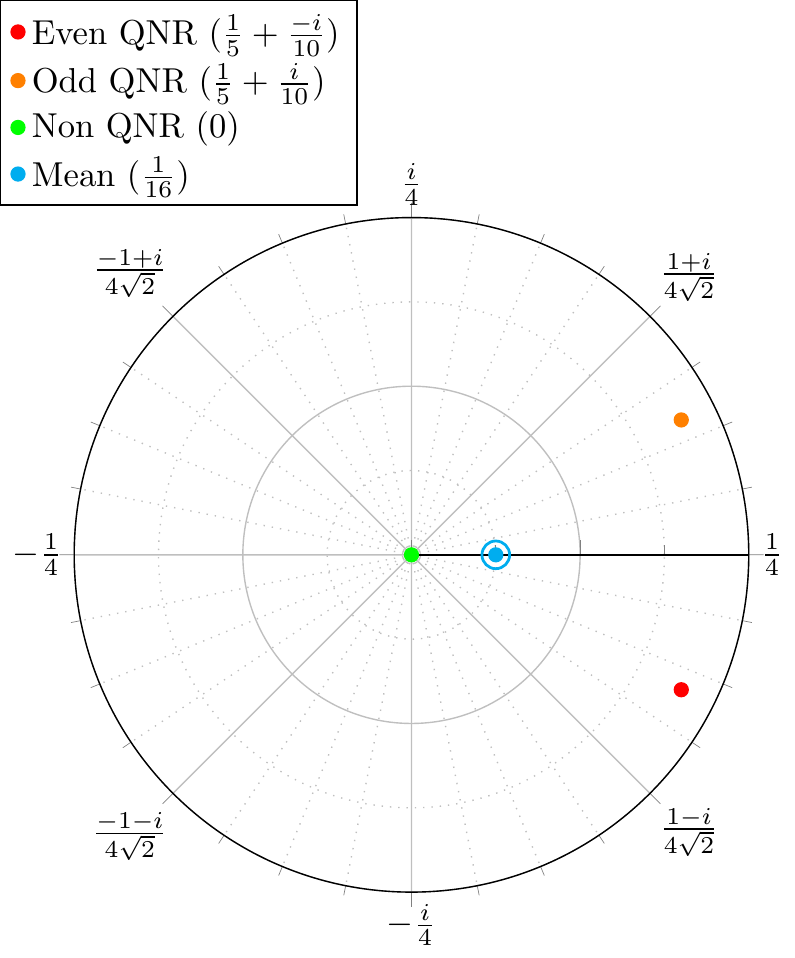}
    \caption{Final position where only QNRs have nonzero (albeit complex) amplitude.}
\label{p41_3}
\end{figure}
\end{enumerate}

After inversion, all of the nonresidues are equally likely to be observed
(Figure~\ref{p41_3}).
Although their phases are different, the magnitudes of their amplitudes are identical, and the magnitude dictates the probability of observation.

\subsection{Computing the angle of rotation}
To find the correct angle of rotation for a given prime, we need to find the angle which moves the mean to half of the length of the superposition vectors.
\[ a\left(\frac{p-1}{2}\right) + 1\cdot\left(2^n-\frac{p-1}{2}\right)= \frac{1}{2}\cdot 2^n\]
Solving for $a$ gives
\[ a = 1- \frac{2^n}{p-1} \]
and thus the angle of rotation is
\begin{align}
  \label{theta}
\theta=\mbox{arccos}\left(1-\frac{2^n}{p-1}\right). 
\end{align}

\subsection{Amplitude before and after inversion}
In an equal superposition of $N=2^n$ states, each state has amplitude $1/\sqrt{N}$. 
Using $\theta$ from equation \eqref{theta},
\begin{align*}
  \cos \theta &= 1-\frac{N}{p-1}\\
  \sin \theta &= \sqrt{1-(1-\frac{N}{p-1})^2}\\
              &= \sqrt{\frac{2N}{p-1}-\frac{N^2}{(p-1)^2}}\\
  &= \frac{\sqrt{N(2p-2-N)}}{p-1}
\end{align*}
After a rotation by $\pm\theta$, the quadratic nonresidues have amplitude
  \[ \frac{1}{\sqrt{N}}\left(\cos \theta \pm i\sin \theta\right). \]

  Since $\theta$ was chosen to move the amplitude mean from $\frac{1}{\sqrt{N}}$ to $\frac{1}{2\sqrt{N}}$, inversion about the mean maps the nonresidue amplitudes to 
\begin{align*}
  &2\frac{1}{2\sqrt{N}}-\frac{1}{\sqrt{N}}\left(\cos \theta\pm i\sin \theta\right)\\
  =&\frac{1}{\sqrt{N}}\left(1-1+\frac{N}{p-1}\mp i\frac{\sqrt{N(2p-2-N)}}{p-1}\right)\\
  =&\frac{\sqrt{N}}{p-1}\mp i\frac{\sqrt{2p-2-N}}{p-1}.
\end{align*}
Since over half of the amplitudes were sent to $0$ in the inversion, the nonresidue amplitude length must have increased.
The new amplitudes of the $\frac{p-1}{2}$ nonresidues have squared length
\begin{align*}
  &\left(\frac{\sqrt{N}}{p-1}\right)^2+\left(\frac{\sqrt{(2p-2-N)}}{p-1}\right)^2\\
  =&\frac{N}{(p-1)^2}+\frac{2p-2-N}{(p-1)^2}\\
  =&\frac{2}{p-1}.
\end{align*}

\section{Algorithmic Complexity}
\subsection{Complexity classes}
Early attempts to find a quantum analog for \xP led to the definition of Exact Quantum Polynomial Time (\xEQP)~\cite{Bernstein97quantumcomplexity}.
Unfortunately, which algorithms \xEQP deemed deterministic polynomial time depended on the finite generating set chosen.
In practice, researchers found \BQP, the quantum analog of \BPP, more useful for discussing what a quantum computer can do efficiently.

In an effort to talk about what a quantum computer can do efficiently and deterministically, the complexity class \EQPk was introduced~\cite{Adleman97quantumcomputability}.
\EQPk allows for controlled unitary gates $U$ on a single qubit where the coefficients of $U$ are from a set $K$.
In practice, the interesting cases come when $K$ is infinite.
For finding quadratic nonresidues, we are interested in the case where $K=\mathbb{C}$.
Note that the coefficients for $U$ may be drawn from $\bar{\mathbb{Q}}$ instead of $\mathbb{C}$~\cite{Adleman97quantumcomputability}.
The discrete logarithm problem over $\mathbb{Z}/p\mathbb{Z}$ is in $\EQP_{\bar{\mathbb{Q}}}$ ~\cite{Mosca04exactquantum}.

Although using an infinite generating set might seem bad, the Solovay-Kitaev theorem~\cite{dawson2005solovaykitaev} proves the single qubit unitary gates can be approximated with exponential precision in polynomial time from a reasonable finite generating set. A more efficient implementation of this idea can be found in \cite{selinger2012efficient}.

We will show that the quadratic nonresidue algorithm is in $\EQP_{\mathbb{C}}$.

\subsection{QC Algorithm for Finding QNRs in Deterministic Polynomial Time}
\label{qc_alg}
Let $M(n)$ be the time required to multiply two $n$ bit numbers. 
The entire algorithm will have complexity equal to the Jacobi symbol, $O(M(n)\log n)=O(n \log^2 n)$~\cite{DBLP:journals/corr/abs-1004-2091,harvey:hal-02070778}.

Given a prime $p\equiv 1\bmod 8$, we define the following for use in algorithm~\ref{quantum_qnr}.
\[ N/2 < p < N=2^n, \]
\[ \theta=\mbox{arccos}\left(1-\frac{N}{p-1}\right), \]
\[ x_0 \text{ is the low bit of } x, \]
\begin{align*}                                                                                                             \label{indicator}
  f(x)&=\left[\jacobi{x}{p}=-1\text{ and } 0\le x<p\right]\\                                                              &=
  \begin{cases}                                                                                                               1&\text{if }x<p\text{ is a quadratic nonresidue},\\                                                                   0&\text{otherwise}.\\
  \end{cases}
\end{align*}
Even though we want to rotate the even QNRs by $\theta$ and the odd QNRs by $-\theta$, for most quantum computers it is likely to be less work to rotate the odd QNRs by $-2\theta$ and then rotate all QNRs by $\theta$.
This is reflected in the algorithm.

\begin{algorithm}
  \caption{
    Deterministic polynomial time algorithm for uniformly sampling quadratic nonresidues of a prime $p\equiv 1 \bmod 8$.
    Let $n, N, \theta, x_0 \text{ and } f(x)$ be defined as in section~\ref{qc_alg}.
}
  \label{quantum_qnr}
  \begin{noindqlist}
  \item $[O(n)]$ Apply $H^{\otimes n}$ to \( \ket{0}^{\otimes n} \) (Hadamard transform).
    \[ \frac{1}{\sqrt{N}}\sum_{x=0}^{N-1} \ket{x} \]
  \item $[O(M(n) \log n)]$ Compute Jacobi symbol indicator.
\[\frac{1}{\sqrt{N}} \sum_{x=0}^{N-1} \ket{x}\ket{\left[\jacobi{x}{p}=-1\right]}\]
\item $[O(n)]$ Compute the indicator for $[x<p]$~\cite{comparator}.
\[\frac{1}{\sqrt{N}} \sum_{x=0}^{N-1} \ket{x}\ket{\left[\jacobi{x}{p}=-1\right]}\ket{[x<p]}\]
  \item $[O(1)]$ Rotate odd QNRs less than $p$ by $-2\theta$, conditioned on $[x<p],\left[\jacobi{x}{p}=-1\right],$ and $x_0$.
  \[\frac{1}{\sqrt{N}} \sum_{x=0}^{N-1} e^{-i2\theta f(x)x_0}\ket{x}\ket{\left[\jacobi{x}{p}=-1\right]}\ket{[x<p]}\]
\item $[O(1)]$ Rotate all QNRs less than $p$ by $\theta$, conditioned on $[x<p]$ and $\left[\jacobi{x}{p}=-1\right]$.
  \[\frac{1}{\sqrt{N}} \sum_{x=0}^{N-1} e^{i\theta f(x)(1-2x_0)}\ket{x}\ket{\left[\jacobi{x}{p}=-1\right]}\ket{[x<p]}\]
\item $[O(M(n)\log n)]$ Uncompute indicator functions.
  \[\frac{1}{\sqrt{N}} \sum_{x=0}^{N-1} e^{i\theta f(x)(1-2x_0)}\ket{x}\]
\item $[O(n)]$ Use a Grover step to invert about the mean $\alpha=\frac{1}{2\sqrt{N}}$.
    \[ \frac{1}{\sqrt N}\sum_{x=0}^{N-1} (1-e^{i\theta f(x) (1-2x_0)})\ket{x} \]
  \item $[O(n)]$ Observe a quadratic nonresidue modulo $p$.
This observation samples uniformly among all quadratic residues modulo $p$.
  \end{noindqlist}
\end{algorithm}

Note that at the end of algorithm~\ref{quantum_qnr}, the quadratic nonresidues have amplitude
\[ \frac{1}{\sqrt N}\left(1-e^{\pm i\theta}\right). \]

As a sanity check, we verify that the nonresidue amplitudes have squared length

\begin{align*}
  &\frac{1}{N}\left( (1-\cos\theta)^2+(\sin \theta)^2\right)\\
  =&\frac{2}{N}\left(1-\cos\theta\right)\\
  =&\frac{2}{N}\left(\frac{N}{p-1}\right)\\
  =&\frac{2}{p-1}.
\end{align*}

And thus the sum of the amplitude lengths squared of the $\frac{p-1}{2}$ nonresidues is $1$.

\section{Additional applications}
Complex rotations can also be used to make Grover's algorithm for a single marked state $\ket{\omega}$ deterministic.
To do this we add one bit to the search space, but allow $\ket{\omega,0}$ and $\ket{\omega,1}$ to both pass our indicator test.
Proceed to do Grover iterations as normal, until we come to the max.
Instead of stepping beyond the max (like a normal Grover iteration would do), we rotate $\ket{\omega,0}$ by $\theta$ and $\ket{\omega,1}$ by $-\theta$, moving the mean to exactly the position that would send all other states to zero after the next inversion about the mean.
Observation then yields $\ket{\omega,0}$ or $\ket{\omega,1}$, both of which have the answer $\omega$.

Grover's algorithm for a marked set, can likewise be made to run in deterministic time, but the output will be a random element of the marked set.

Cubic nonresidues and higher can also be found in deterministic time with this approach.

Furthermore, a random element of any set of known size can be produced, assuming you have a function which identifies elements of the set.

Interestingly, finding primitive roots is still not known to be in \EQPC because the test identifying primitive roots of $p$ requires factoring $p-1$, and factoring is not known to be in \EQPC.
\section{Conclusion}
Although finding quadratic nonresidues is easy in practice, being able to generate them in deterministic time is something new.
Interestingly, since the distribution of quadratic nonresidues is not completely understood, it may be useful in evaluating NISQ devices.
Essentially, if an approximation of this algorithm creates a distribution from a single Jacobi symbol calculation that has a better sampling distribution for finding quadratic nonresidues than any classical algorithm limited to a single Jacobi symbol calculation, it may present evidence that the NISQ device is doing something quantum.
This idea is explored further in \cite{nisqqnr}.

\bibliographystyle{IEEEtran}  
\bibliography{QNR_EQP}  

\end{document}